\documentclass[letter, 10pt, conference]{ieeeconf}      

\IEEEoverridecommandlockouts                              %
\usepackage{array,multicol,multirow,booktabs,longtable,tabularx}
\usepackage{subfig}
\usepackage{amssymb,mathtools,dsfont,bbm,cuted,stmaryrd,bm}
\usepackage{xspace,algorithmic,algorithm}

\usepackage{tikz}
\usetikzlibrary{math,arrows.meta,spy,automata,intersections,decorations,arrows,decorations.pathmorphing,decorations.markings,decorations.footprints,fadings,calc,trees,mindmap,shadows,decorations.text,patterns,positioning,shapes,matrix,fit,3d,plotmarks,svg.path}

\usepackage{pgfplots}

\usepgfplotslibrary{patchplots}

\usepackage{tikz-3dplot}
\usepackage{tkz-euclide}

\usepackage{placeins}
\usepackage{cuted}

\usepackage{soul}

\let\labelindent\relax
\usepackage{enumitem}

\newcommand{\be}{\begin{equation}}
\newcommand{\ee}{\end{equation}}

\def\bal#1\eal{\begin{align}#1\end{align}}
\def\baln#1\ealn{\begin{align*}#1\end{align*}}

\newcommand{\ben}{\begin{equation*}}
\newcommand{\een}{\end{equation*}}

\newcommand{\re}[1]{\mathbb{R}^{#1}}%

\newcommand{\bbm}{\begin{bmatrix}}
\newcommand{\ebm}{\end{bmatrix}}

\newcommand{\bBm}{\begin{Bmatrix}}
\newcommand{\eBm}{\end{Bmatrix}}

\newcommand{\bvm}{\begin{vmatrix}}
\newcommand{\evm}{\end{vmatrix}}

\newcommand{\bVm}{\begin{Vmatrix}}
\newcommand{\eVm}{\end{Vmatrix}}

\newcommand{\bpm}{\begin{pmatrix}}
\newcommand{\epm}{\end{pmatrix}}

\newcommand{\bnm}{\begin{matrix}}
\newcommand{\enm}{\end{matrix}}

\newcommand{\bi}{\begin{itemize}}
\newcommand{\ei}{\end{itemize}}

\newcommand{\bse}{\begin{subequations}}
\newcommand{\ese}{\end{subequations}}

\newcommand{\secref}[1]{Section~\ref{#1}\xspace}

\newcommand{\remref}[1]{Remark~\ref{#1}\xspace}

\newenvironment{proof-sketch}{\noindent{ \textit{Sketch of Proof}:}\hspace*{0.5em}}

\DeclareMathOperator{\diag}{diag}
\DeclareMathOperator{\real}{Re}

\usepackage{amsthm}

\theoremstyle{plain}

\newtheorem{prop}{Proposition}
\newtheorem{cor}{Corollary}
\newtheorem{lem}{Lemma}

\newtheorem{rem}{Remark}

\newcommand{\eor}{\ensuremath{\hfill\blacklozenge}}
\def\qed{\hfill\vrule height 1.6ex width 1.5ex depth -.1ex}

\DeclareMathOperator{\vol}{Vol}

\title{\LARGE \bf Tight displacement-based formation control under bounded disturbances. A set-theoretic perspective.}

\author{Vlad-Matei Anghelu\c t\u a$^{1}$, Bogdan Gheorghe$^{1}$, Daniel Ioan$^{1}$, Ionela Prodan$^{2}$, Florin Stoican$^{1}$
\thanks{$^{1}$Vlad-Matei Anghelu\c t\u a, Bogdan Gheorghe, Daniel Ioan, and Florin Stoican are with Dept. of Automatic Control and Systems Engineering, RePlan team, CAMPUS Research Institute, Politehnica Bucharest, Romania 
        {\tt\small vlad.angheluta@stud.acs.upb.ro, \{bogdan.gheorghe1807, daniel\char`_mihail.ioan, florin.stoican\}@upb.ro}}%
\thanks{$^{2}$ Ionela Prodan is with Univ. Grenoble Alpes, Grenoble INP$^\dagger$, LCIS, F-26000, Valence, France. $^\dagger$Institute of Engineering and Management, Univ. Grenoble Alpes
        {\tt\small ionela.prodan@lcis.grenoble-inp.fr}}%
}

\usepackage{graphicx}      

\usepackage{amsmath,amsfonts,amssymb,amsthm}
\usepackage{xcolor}

\usepackage[algo2e, linesnumbered, algoruled]{algorithm2e}

\usepackage{hyperref}

\DeclareMathOperator*{\argmin}{arg\,min}

\newcommand{\rk}[1]{\textrm{rank}(#1)}
\newcommand{\laplacian}{\mathcal{L}}

\begin{document}

\maketitle

\begin{abstract}

This paper investigates the synthesis of controllers for displacement-based formation control in the presence of bounded disturbances, specifically focusing on uncertainties originating from measurement noise. While the literature frequently addresses such problems using stochastic frameworks, this work proposes a deterministic methodology grounded in set-theoretic concepts. By leveraging the principles of set invariance, we adapt the theory of ultimate boundedness to the specific dynamics of displacement-based formations. This approach provides a rigorous method for analyzing the system's behavior under persistent disturbances. Furthermore, this set-theoretic framework allows for the optimized selection of the proposed control law parameters to guarantee pre-specified performance bounds. The efficacy of the synthesized controller is demonstrated in the challenging application of maintaining tight formations in a multi-obstacles environment. 
\end{abstract}




\section{Introduction}

The increasing prevalence of real-world applications involving multiple unmanned vehicles has intensified the interest of the control community in developing reliable and efficient formation control strategies~\cite{wang_affine_2021}, aimed at enhancing safety, reliability, and procedural efficiency. Various approaches have been extensively investigated, typically balancing the sensing and communication capabilities of the agents within the formation~\cite{oh_survey_2015}. Among them, consensus-based displacement formation control~\cite{olfati-saber_consensus_2007} stands out for its ability to coordinate agents to form and maintain a desired geometric pattern defined by relative position vectors (displacements) between neighbors~\cite{ji_distributed_2006}. This method is widely adopted due to its linear simplicity, flexibility, and robustness, ensuring convergence to the target formation~\cite{chen_displacement-based_2023}. While model selection can introduce numerical and computational challenges, simplified dynamics such as single~\cite{cao_generalized_2008} or double~\cite{sun_rigid_2017} integrator models are often preferred. Although the former eases analysis and implementation, we adopt the latter as a more realistic choice for robotic systems.


However, the idealized simplicity of this control strategy is often challenged by real-world physical constraints. The system’s performance is susceptible to disruptions such as measurement noise, which corrupts state estimation~\cite{huang_coordination_2009}, and communication disturbances that interfere with the control process~\cite{tang_formation_2019}. A key challenge, therefore, lies in designing reliable control laws that explicitly account for these uncertainties. This necessitates efficient uncertainty modeling, typically through either deterministic bias representations~\cite{cao_formation_2011} or stochastic noise formulations~\cite{huang_coordination_2009}.



For bounded disturbances, the formation's trajectories can be constrained via invariant sets at steady state. Exact reachability methods, which propagate an initial seed set through the system dynamics, exhibit exponential complexity and quickly become intractable in higher dimensions~\cite{blanchini_set_1999}. Invariant-set approaches shift the computational burden offline but remain challenging. Representation-wise, polyhedral formulations~\cite{rakovic_minimal_2005} are accurate yet computationally demanding, whereas ellipsoidal representations~\cite{khalil2002nonlinear}, though numerically robust, impose structural constraints that often lead to conservative results. Ultimate bounds constructions~\cite{kofman_control_2008,kofman_computation_2014} offer a balanced alternative—computationally efficient while maintaining acceptable accuracy. Building on these insights, the main contributions of this paper are as follows:
\begin{enumerate}[label=\roman*)]
    \item computation of the Jordan canonical decomposition of the formation Laplacian;
    \item derivation of the volume of the ultimate bounds set that characterizes the steady-state tracking error dynamics;
    \item formulation and solution of an optimization problem yielding a tight and safe displacement-based formation.
\end{enumerate}
All results are developed under the assumption of a proportional–derivative control law, with its tuning parameters appearing explicitly in each of the aforementioned items.



The remainder of the paper is organized as follows. Section~\ref{sec:prerequisites} introduces the prerequisites on ultimate bounds and graph-based formations. Section~\ref{sec:main_idea} defines the ultimate bounds for the displacement-based formation and designs a stable tight formation. Section~\ref{sec:results} validates the theoretical results, and Section~\ref{sec:conclusions} concludes the paper.
\subsection*{Notations}
$O_{n\times m}\in \mathbb R^{n \times m}$ is the matrix whose entries are zero, with the shorthand notation $O_{n}$, whenever $n = m$. $\mathbf 0_n \in \mathbb R^{n}$ is the vector with zero entries. $I_{n} \in \mathbb R^{n\times n}$ is the identity matrix, and $\mathbf 1_n  \in \mathbb R^{n}$ denotes the vector of $n$ values of one. For an arbitrary matrix $A$, `$\mathrm{eig}(A)$' defines the vector of eigenvalues of the matrix, `\emph{rank(A)}' its rank, and $ker(A)$ its kernel space. The span of a vector $v$ is given by $span\{v\}$. The symbols `$\odot$' and `$\otimes$' denote the Hadamard and Kronecker product, respectively. For a vector $x \in \mathbb R^{n}$, its infinity norm is given as $\|x\|_{\infty} := \max(|x_1|, |x_2|, \ldots, |x_n|)$. The Minkowski sum of two sets is defined as $X\oplus Y=\{x+y:\: \forall x \in X, \forall y\in Y\}$.

\section{Prerequisites}
\label{sec:prerequisites}
We adapt the ultimate bound constructions from~\cite{kofman_systematic_2007} to derive closed-form expressions for both the invariant set associated with the formation dynamics under displacement control and its corresponding volume.
\subsection{Set invariance}
\label{sec:set_inv}
Consider a continuous-time linear time-invariant (LTI) system governed by
\begin{equation}
\label{eq:linear_dynamics}
    \dot{x} = A x + \delta,
\end{equation}
where $A \in \mathbb{R}^{n \times n}$ is a Hurwitz matrix and $\delta \in \Delta \subset \mathbb{R}^n$ is a bounded disturbance. Let the disturbance set be zonotopic, $\Delta = \langle c, G \rangle = \{\, c + G \lambda : \|\lambda\|_\infty \le 1 \,\}$, with center $c \in \mathbb{R}^n$ and generator matrix $G \in \mathbb{R}^{n \times D}$, as defined in~\cite{fukuda_polyhedral_nodate}. Furthermore, recall that a set $\Omega$ is said to be \emph{robust positively invariant} (RPI) if it satisfies the inclusion $A\Omega \oplus \Delta \subseteq \Omega$,~\cite{blanchini_set_1999}. In this context, a slightly modified form of the \emph{ultimate bounds} construction from~\cite[Thm.~1]{kofman_systematic_2007} is obtained.
\begin{prop}
\label{prop:ub}
    Let $A = V \Lambda V^{-1}$ be the Jordan canonical decomposition of $A$. Then, the sets 
    \begin{align}
    \label{eq:ub}
        \Omega_{UB} &= \bigl\{\, x \in \mathbb{R}^n : \left|V^{-1}x\right| \leq b  \bigr\},\\
        \label{eq:ub_box}B_{UB} & =\bigl\{\, x \in \mathbb{R}^n : \left|x\right| \leq \left|V\right|b  \bigr\},
    \end{align}
    with notation 
    \begin{equation}
    \label{eq:notation_rhs_ub}
        b=\bigl|\bigl[\real(\Lambda)\bigr]^{-1}\bigr|
        \cdot \left(|V^{-1}c| + |V^{-1}G|\cdot \mathbf 1_D\right),
    \end{equation} 
    are robust positively invariant (RPI) for the dynamics~\eqref{eq:linear_dynamics}.
\end{prop}

\begin{proof}
Left-multiplying~\eqref{eq:linear_dynamics} by $V^{-1}$ and defining $y = V^{-1}x$ yields the canonical form 
$\dot{y} = \Lambda y + V^{-1}\delta$. 
Applying~\cite[Lemma~3]{kofman_systematic_2007} gives 
$|y| \leq \bigl|\bigl[\real(\Lambda)\bigr]^{-1}\bigr| 
\max_{\delta \in \Delta} |V^{-1}\delta|$.
Since $\Delta = \langle c, G \rangle$, it follows that
\begin{multline*}
    \max_{\delta \in \Delta} |V^{-1}\delta|
    = \max_{\|\lambda\|_\infty \leq 1} |V^{-1}(c + G\lambda)|\\
    \leq |V^{-1}c| + |V^{-1}G|\cdot \mathbf 1_D,
\end{multline*}
which directly leads to~\eqref{eq:ub}. Noting that $|x| =|V\cdot V^{-1}x| = |V|\cdot |V^{-1}x|$ leads to \eqref{eq:ub_box}, concluding the proof.\qed
\end{proof}
\begin{lem}
    The set \eqref{eq:ub} may be equivalently written as a zonotope:
    \begin{equation}
        \label{eq:ub_zonotope}
        \Omega_{UB}=\left\langle \mathbf 0_n, V \odot \left(\mathbf 1_n \otimes b^\top\right)\right\rangle
    \end{equation}
\end{lem}
\begin{proof}
    Taking $\left|V^{-1}x\right| \leq b$ from \eqref{eq:ub}, we write $-b\leq V^{-1}x\leq b$, which, for the $i$-th index corresponds to $\mp V_i^\top b\leq x_i\leq \pm V_i^\top b$. This is equivalent with having $x_i = V_i^\top b \cdot \lambda_i$ where $|\lambda_i|\leq 1$. Repeating this for all indices $i$ and putting in matrix form leads to \eqref{eq:ub_zonotope}, concluding the proof. \qed
\end{proof}

\noindent Applying \cite[Cor. 3.4]{gover_determinants_2010} to \eqref{eq:ub_zonotope}, we obtain the next result. 
\begin{cor}
    The volume of \eqref{eq:ub_zonotope} is given by
    \begin{equation}
    \label{eq:vol_ub}
        \vol\bigl(\Omega_{UB}\bigr)= \left(2\det V \cdot  \prod_{i=1}^n b_i\right)^2.
    \end{equation}
\end{cor}
\begin{proof}
As shown in~\cite[Cor.~3.4]{gover_determinants_2010}, the volume of a zonotope of the form $\langle 0, G \rangle \subset \mathbb{R}^n$, with $G \in \mathbb{R}^{n \times n}$, is given by $4\det(G^\top G)$. Taking $G$ as in~\eqref{eq:ub_zonotope} and applying determinant properties yields $\det G = \det\!\left[V \odot \left(\mathbf{1}_n \otimes b^\top\right)\right] = \left(\prod_{i=1}^n b_i\right)\det V$. Since $\det(G^\top G) = (\det G)^2$, substituting the expression above gives~\eqref{eq:vol_ub}, concluding the proof. \qed
\end{proof}

\subsection*{Illustrative example}

We present an example inspired by~\cite{kofman_non-conservative_2005} of computed ultimate bounds for a particular system and illustrate the robust positive invariance of the obtained sets $\Omega_{UB}$ and $B_{UB}$. Thus we take a system described by~\eqref{eq:linear_dynamics}, where:
\begin{equation*}
A = \begin{bmatrix}
    0 & -2 \\
    1 & -3
    \end{bmatrix}, \quad |\delta| \leq \begin{bmatrix}
                                      0.2 \\
                                      0.2
                                      \end{bmatrix}
\end{equation*}

$A$ is Hurwitz with $\mathrm{eig}(A)=\{-1,-2\}$. By applying~\eqref{eq:notation_rhs_ub}, we obtain $b^{\top}=\begin{bmatrix} 1.1 & 0.7\end{bmatrix}$. The volume of $\Omega_{UB}$, computed using~\eqref{eq:vol_ub}, is $\vol\bigl(\Omega_{UB}\bigr) = 0.23$.

Next we take a random collection of 5 points inside $B_{UB}$ and simulate the trajectories over a time horizon of $T = 5s$.
\begin{center}
\includegraphics[width=\columnwidth]{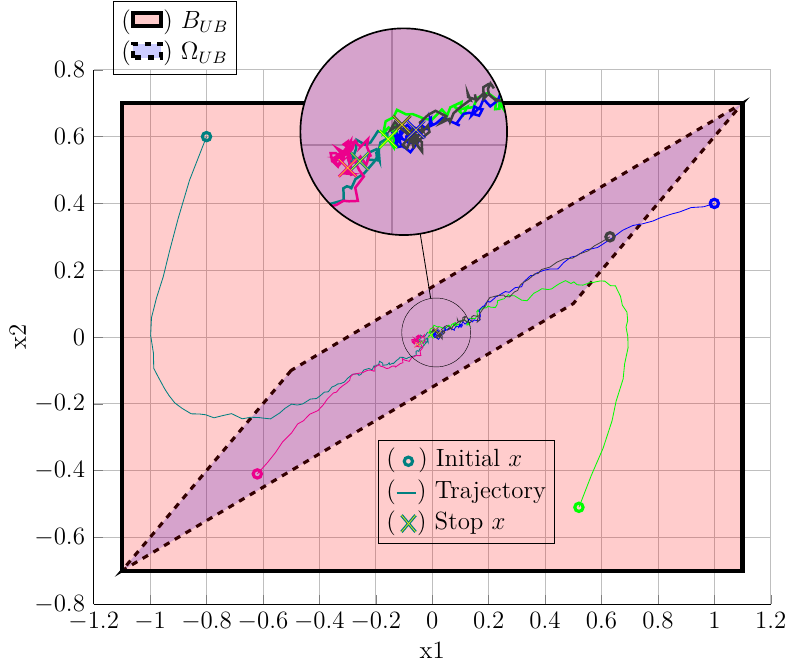}
\end{center}

The outer rectangle represents the box set $B_{UB}$ and the inner trapezoid represents the less conservative set $\Omega_{UB}$. As seen above, all trajectories remain in $\Omega_{UB}$. Furthermore, those which begin outside the $\Omega_{UB}$ set, eventually converge inside of it as described by~\cite[Thm.~2]{kofman_nonlyapunov_nodate}.

\subsection{Graph-based formation descriptions}
\label{sec:graphs}


The interaction topology of a multi-agent system comprising $N$ agents is conventionally modeled by an undirected graph $\mathcal{G} = (\mathcal{V}, \mathcal{E})$. The set of vertices, $\mathcal{V} = \{1,\ldots, N\}$, represents the agents, while the edge set, $\mathcal{E} \subseteq \mathcal{V} \times \mathcal{V}$ with cardinality $M = |\mathcal{E}|$, defines the {interactions} links among them. For each $i \in \mathcal{V}$ we define $\mathcal{N}_i=\{j \in \mathcal{V} \mid (i,j) \in \mathcal{E}\}$ as set of its neighbors.  

By assigning an orientation to each edge\footnote{ We index each edge in $\mathcal{E}$  by $k(i,j)$. Whenever it is clear, we drop the arguments and refer to the edge by '$k$'.} $k = k(i,j)  \in \{1,\ldots, M\}$, the topology is formally captured by the incidence matrix, \cite{dimarogonas_stability_2010}, $H = [h_{ki}] \in \mathbb{R}^{M \times N}$ whose elements are defined as follows: 
\begin{equation*}
h_{ki} = \left\{\begin{array}{ll}
\hphantom{-}1, & \textrm{if node $i$ is the sink (head) of edge $k$},\\
-1, & \textrm{if node $i$ is the source (tail) of edge $k$},\\
\hphantom{-}0, & \textrm{otherwise}.
\end{array} \right.
\end{equation*} 

The Laplacian of a connected\footnote{Connectivity is inherently and necessarily a prerequisite for successful formation control.} graph $\mathcal{G}$ is the symmetric positive semidefinite matrix $\laplacian = H^\top H$ \cite[Lemma 1]{chen_displacement-based_2023}. Moreover, we have: $\rk{\laplacian} = N-1$ and $\ker(\laplacian)= \ker (H)=\textrm{span}\{\mathbf{1}_N\}$. 

 If each agent $i$ is located in $p_i (t) \in \re{n}$ at time $t \geq 0$, then $p=\bbm p_1^\top  & \ldots & p_n^\top\ebm^\top \in \re{nN}$ is a \textit{representation} of $\mathcal{G}$. In what follows, we consider the case of double-integrator agents:
\begin{equation}
    \label{eq:LTI_double_int}
    {\bbm \dot p \\ \dot v\ebm} = \underbrace{\bbm O_{nN} & I_{nN} \\ O_{nN} & O_{nN}\ebm}_{A} \bbm p \\ v\ebm + \underbrace{\bbm O_{nN}\\ I_{nN}\ebm}_{B} u
\end{equation}

where $v \in \re{nN}$ is the velocity vector, while the inputs $u \in \re{nN}$ are the agents' accelerations. We assume that the agents are able to sense the relative positions and velocities of their neighbors w.r.t. to a global coordinate system. Since we deal with  a displacement-based formation approach, we define: $ z_{k(i,j)}= z_{ij} = p_i - p_j$ -  the relative position of the $i$-th agent  w.r.t $j$-th one. Compactly, we have $z=\bbm z_1^\top & \ldots & z_M^\top \ebm^\top  \in \re{nM}$ the relative position vector, that can be further written as $z= \tilde{H}p$ with $\tilde{H}=(H \otimes I_{n})$. Similarly, we have $ \zeta_{k(i,j)}= \zeta_{ij} = v_i - v_j$ and  $\zeta \in \re{nM}$ the relative velocity vector, $\zeta=\tilde{H}v$. The desired formation is specified by the following conditions: $z \mapsto z^\star$ and $\zeta \mapsto \zeta^\star$, therefore a control law for the agents can be designed as follows
\begin{equation}
\label{eq:control_law}
    u_i = -k_p \sum_{j\in \mathcal{N}_i} (z_{ij} -z^\star_{ij}) -k_v \sum_{j\in \mathcal{N}_i} (\zeta_{ij} -\zeta^\star_{ij})
\end{equation}
\noindent where $k_p$ and $k_v$ are strictly positive scalar. Based on \eqref{eq:control_law} and on the properties of the interaction graph, we obtain the tracking error dynamics: 

\begin{equation}
\label{eq:tracking_error}
\begin{bmatrix}
    \dot{e}_p \\
    \dot{e}_v
\end{bmatrix} = \underbrace{\begin{bmatrix} O_{nN} & I_{nN}\\
                                -k_p(\laplacian \otimes I_n) & -k_v(\laplacian \otimes I_n) \end{bmatrix}}_{\Tilde{\Gamma}} \begin{bmatrix}
    e_p \\
    e_v
\end{bmatrix}.
\end{equation}

Since we assume that each agent can measure the relative positions of its neighbors, it is clear that the main source of disturbance is coming from these measurements. Therefore we consider, as in \cite{chen_displacement-based_2023}, an additive noise at each $t \geq 0$ on relative position and velocity vectors: $\tilde{z}_{ij}(t) = z_{ij}(t) + \varepsilon_{ij} (t)$ and, respectively, $\tilde{\zeta}_{ij}(t) =\zeta_{ij} (t) + \xi_{ij}(t)$. Thus, the dynamics \eqref{eq:tracking_error} become: 
\begin{align}
    \label{eq:noisy_tracking_error}
    \begin{bmatrix}
    \dot{e}_p \\
    \dot{e}_v
\end{bmatrix} = \Tilde{\Gamma}\begin{bmatrix} 
    e_p \\
    e_v
\end{bmatrix}
+ \tilde H \begin{bmatrix} 
    \varepsilon \\
    \xi
\end{bmatrix}
\end{align}

\noindent where $\varepsilon$ and $\xi$ are bounded disturbances in $\re{n}$ (i.e., $\varepsilon_{\min} \leq \varepsilon \leq \varepsilon_{\max}$ and $\xi_{\min} \leq \xi \leq \xi_{\max}$). Thus, we have dynamics similar to \eqref{eq:linear_dynamics} and, with some small modifications (see \remref{rem:modified_laplacian}), the results from \ref{sec:set_inv} can be employed. 

\subsection*{Illustrative example}
\label{sec:illustrative_example_LFF}
Consider an acyclic LFF (\textit{Leader-First-Follower}) formation \cite{cao_control_2008} composed by:  agent $i=1$ - the \textit{leader}, agent $i=2$ - the \textit{first-follower} and $N-1$ agents ($i \geq 3 $) - the ordinary followers. Thus, the interaction graph has $M=1+2(N-2)$ edges, resulting the next sets of neighbors (choosing a standard edge orientation for the formation):
\begin{equation}
\label{eq:lff_neigh}
    \mathcal{N}_1 =\emptyset,\: \mathcal{N}_2 =\{1\},\: \mathcal{N}_i =\{i-2, i-1\}, \forall i \geq 3.    
\end{equation}

For $N=4$ we can construct the matrices of interest: 
\begin{equation}
    \mkern-16mu H\mkern-4mu =\mkern-4mu \begin{bmatrix}
        1 & -1 & \hphantom{-}0  & \hphantom{-}0 \\ 1 & \hphantom{-}0  & -1 & \hphantom{-}0 \\ 0 & \hphantom{-}1 & -1 & \hphantom{-}0 \\ 0 & \hphantom{-}1 & \hphantom{-}0  & -1 \\
        0 & \hphantom{-}0 & \hphantom{-}1 & -1
    \end{bmatrix}\mkern-4mu,  \laplacian\mkern-4mu  =\mkern-4mu \begin{bmatrix}
    \hphantom{-}2 & -1 & -1 & \hphantom{-}0 \\
    -1 & \hphantom{-}3 & -1 & -1 \\
    -1 & -1 & \hphantom{-}3 & -1 \\
    \hphantom{-}0 & -1 & -1 & \hphantom{-}2
\end{bmatrix}.
\end{equation}

\begin{rem}
\label{rem:modified_laplacian}
    Due to its properties, the Laplacian has a $0$ eigenvalue, which renders a non-Hurwitz $\Tilde{\Gamma}$ in \eqref{eq:noisy_tracking_error}. This  prevents the direct application of set invariance theory from \secref{sec:set_inv}. To resolve this issue, we can augment the control law of the leader (conventionally, $u_1$) by considering an additional term in \eqref{eq:control_law}, i.e. $-k_p \alpha (p_{1} - p_{1}^\star)$ with $p_1^\star$ coming from a desired trajectory of the leader. This augmentation is equivalent with considering a bias term $\alpha$ in the first element of  $\laplacian$ and, by using this small modification, it ensures  all eigenvalues positive and, implicitly,  a Hurwitz $\Tilde{\Gamma}$. In what follows, the notation $\laplacian$ represents the modified Laplacian. \eor
\end{rem}

\section{Main idea}
\label{sec:main_idea}

We have previously introduced the concepts of set invariance in the context of ultimate bounds and displacement-based formations. We now employ the former to analyze the steady-state behavior of the latter. As a first step, the next section establishes a connection between the closed-loop tracking error dynamics and their corresponding \emph{ultimate bounds} invariant set.

\subsection{Tracking error dynamics invariant set}
For subsequent use, we recall the spectral properties of the Kronecker product~\cite[Thm.~4.2.12]{horn1994topics}, which hold for any square matrices $A$ and $B$ with corresponding eigenvalue–eigenvector pairs $(\xi, x)$ and $(\zeta, y)$ satisfying $Ax = \xi x$ and $By = \zeta y$:
\begin{enumerate}[label=\roman*)]
    \item\label{item:k1} $(A \otimes B)(x \otimes y) = (\xi \zeta)(x \otimes y)$,
    \item\label{item:k2} $\mathrm{eig}(A \otimes B) = \mathrm{eig}(A) \otimes \mathrm{eig}(B)$;
\end{enumerate}
together with the modulus property
\begin{enumerate}[label=\roman*),resume]
    \item\label{item:k3} $|A\otimes I_N|=|A|\otimes I_N$.
\end{enumerate}
\begin{prop}
\label{prop:factorization}
Let $\laplacian = V \Lambda V^{-1}$ denote the Jordan canonical decomposition of the Laplacian, as defined in \secref{sec:graphs}. 
Then, the Jordan canonical decomposition of $\tilde\Gamma$, the state matrix in~\eqref{eq:noisy_tracking_error}, 
is given by $\tilde\Gamma = \tilde V_\Gamma \tilde\Lambda_\Gamma \tilde V_\Gamma^{-1}$, whose components are defined as follows:
\begin{subequations}
    \begin{align}
    \label{eq:eigenvalues_gamma}
    \tilde\Lambda_\Gamma &=\begin{bmatrix}\Lambda_+&O_N\\ O_N & \Lambda_-\end{bmatrix}\otimes I_n,\\
        \label{eq:eigenvectors_gamma}\tilde V_{{\Gamma}}&=
    \begin{bmatrix}
    V & V \\
    V\Lambda_+ & V\Lambda_-
    \end{bmatrix}\otimes I_n,
    \end{align}
    \end{subequations}
where $\Lambda_\pm=\mathrm{diag}\bigl(\{\mu_{i,\pm}\}_{i\in \{1, \ldots, N\}}\bigr)$ with:
\begin{equation}
\label{eq:mu_i}
    \mu_{i,\pm}
=\frac{-\,k_v\lambda_i \pm \sqrt{(k_v\lambda_i)^2-4k_p\lambda_i}}{2} ,\: \lambda_i \in \mathrm{eig}(\laplacian).
\end{equation}
\end{prop}
\begin{proof}
Recalling the state matrix $\tilde\Gamma$ from~\eqref{eq:noisy_tracking_error} and applying the Kronecker product's definition, we obtain
\begin{equation*}
\label{eq:gamma}
\Tilde{\Gamma} =  \underbrace{\begin{bmatrix}
\hphantom{-}O_N & \hphantom{-}I_N\\[2mm]
-k_p \laplacian & -k_v \laplacian
\end{bmatrix}}_{\Gamma} \otimes I_n 
\end{equation*} 
Left- and right-multiplying $\Gamma$ by $\left(I_2 \otimes V\right)$ and its inverse gives
\begin{equation*}
\left(I_2 \otimes V\right)^{-1}{\Gamma}\left(I_2 \otimes V\right)
=
\begin{bmatrix}
\hphantom{-}O_N & \hphantom{-}I_N\\
-k_p \Lambda & -k_v \Lambda
\end{bmatrix},
\end{equation*}
which, after appropriate row and column permutations, decomposes into $2\times2$ blocks of the form
\begin{equation*}
\begin{bmatrix}
0 & 1\\
-k_p \lambda_i & -k_v \lambda_i
\end{bmatrix}, \quad \lambda_i \in \mathrm{eig}(\laplacian),
\end{equation*}
with corresponding eigenvalues \eqref{eq:mu_i}. Since all transformations applied to ${\Gamma}$ are similarity transformations, its spectrum is preserved. Moreover, via property~\ref{item:k2}, it follows that  
\begin{equation*}
\mathrm{eig}(\tilde\Gamma) = \mathrm{eig}({\Gamma} \otimes I_n) = \mathrm{eig}({\Gamma}) \otimes \mathbf{1}_n,
\end{equation*}
which, together with \eqref{eq:mu_i}, directly leads to~\eqref{eq:eigenvalues_gamma}. 

\noindent For the next step, we consider
\begin{equation}
    \mathbf{w}_{i,\pm}
    =
    \begin{bmatrix}
    \mathbf{v}_i\\
    \mu_{i,\pm}\,\mathbf{v}_i
    \end{bmatrix},
    \quad \mathrm{where} 
     \quad \laplacian \mathbf{v}_i=\lambda_i\mathbf{v}_i,
\end{equation}
and using~\eqref{eq:mu_i}, it follows that 
${\Gamma} \mathbf{w}_{i,\pm} = \mu_{i,\pm} \mathbf{w}_{i,\pm}$.
By definition, $\mathbf{w}_{i,\pm}$ are therefore the eigenvectors of the matrix ${\Gamma}$. Taking $A=\Gamma$ and $B=I_n$ in property~\ref{item:k1} we obtain that $(\Gamma\otimes I_n)(\mathbf w_{i,\pm}\otimes \mathbf e_j) = (\mu_{i,\pm}\otimes 1)(\mathbf w_{i,\pm}\otimes \mathbf e_j)= \mu_{i,\pm}\cdot (\mathbf w_{i,\pm}\otimes \mathbf e_j)$. Iterating for all indices $i,j$ and stacking horizontally the resulting column vectors gives the matrix of eigenvectors \eqref{eq:eigenvectors_gamma}, thus concluding the proof.\qed
\end{proof}
\begin{cor}
\label{cor:volume_proof}
    The volume of the ultimate bounds \eqref{eq:ub} for matrix \eqref{eq:gamma} is given by
    \begin{multline}
    \label{eq:volume_ub}
        \vol\Omega_{\Gamma,\mathrm{UB}}(k_p,k_v)=\\ \left[[\mathrm{det}V]^{2n}\prod\limits_{i=1}^Nd_i^2\biggl(\dfrac{1+\lambda_i(k_p+k_v)}{\sqrt{\Delta_i}k_p\lambda_i}\biggr)^n\right]^2
    \end{multline}
    where $\Delta_i=k_v^2\lambda_i^2-4k_p\lambda_i$ and $d = V^{-1}\overline{\delta}_N$.
\end{cor}
\begin{proof}The idea is to express both $\tilde V_\Gamma$ and $\tilde \Lambda_\Gamma$ explicitly in terms of the control gains $k_p,k_v$. The straightforward but rather convoluted details are illustrated in appendix ~\ref{proof:cor_volume}.
    \qed
\end{proof}



\subsection{Tight formation computation}

The gain pair $(k_p, k_v)$, when substituted into the matrices $\Lambda_\Gamma$ and $V_\Gamma$ that factorize $\Gamma$ as in Prop.~\ref{prop:factorization}, yields-via Prop.~\ref{prop:ub}-the ultimate bounds set $\Omega_{\Gamma,\mathrm{UB}}$ and the hyper-rectangle $B_{\Gamma,\mathrm{UB}}$ enclosing the steady-state tracking error dynamics~\eqref{eq:tracking_error}. Their computation is summarized in the following result.
\begin{lem}
    \label{lem:control_synthesis}
    The gains $k_p>0$ and $k_v>0$ that guarantee closed-loop stability and minimize the volume of the invariant set bounding the tracking error dynamics~\eqref{eq:tracking_error}, computed as in~\eqref{eq:volume_ub}, are obtained by solving:
    \begin{subequations}
    \begin{align}
        \label{eq:control_a}
        (k_p^\star, k_v^\star) &= \arg\min_{k_p,\,k_v}\; \vol\bigl(\Omega_{\Gamma,\mathrm{UB}}(k_p,k_v)\bigr)\\[2mm]
        \label{eq:control_b}
        \text{s.t.}\quad & k_v^2\lambda_i^2 - 4k_p\lambda_i \ge 0,\\[1mm]
        \label{eq:control_c}
        & \underline{\mu} \le -k_v\lambda_i + \sqrt{k_v^2\lambda_i^2 - 4k_p\lambda_i} \le \overline{\mu},
    \end{align}
    \end{subequations}
    for all $\lambda_i \in \mathrm{eig}(\laplacian)$ and given bounds $\underline{\mu} < \overline{\mu} < 0$.
\end{lem}
\begin{proof}
Conditions~\eqref{eq:control_b} guarantee that the eigenvalues $\mu_{i,\pm}$, defined in~\eqref{eq:mu_i}, are real. Subsequently, constraint~\eqref{eq:control_c} enforces stability by confining them to the interval $[\underline{\mu},\, \overline{\mu}]$, which, by construction, lies within $\mathbb{C}^-$. Finally, the cost term~\eqref{eq:control_a}, which penalizes the volume of the ultimate bounds set $\Omega_{\Gamma,\mathrm{UB}}$ defined in~\eqref{eq:volume_ub}, completes the proof.
\qed
\end{proof}

Applying Lem.~\ref{lem:control_synthesis} enables the formulation of the so-called \emph{steady-state tight-formation} problem, in which the reference position displacements $z_{ij}^\star$ are freely chosen provided the following competing objectives are satisfied:
\begin{enumerate}[label=(\roman*)]
    \item\label{item:opt-1} all relative velocities are zero, $ \zeta =0$;
    \item\label{item:opt-2} the leader tracks the desired target, $p_1^\star = \bar{p}$;
    \item\label{item:opt-3} all followers remain as closely grouped as possible around the leader (the ``tight'' requirement);
    \item\label{item:opt-4} the formation avoids collisions with known obstacles, modeled as a union of polyhedra, $\mathcal O  = \bigcup_\ell S_\ell$.
\end{enumerate}

\begin{figure*}[ht]
  \centering
  \includegraphics[width=2.07\columnwidth]{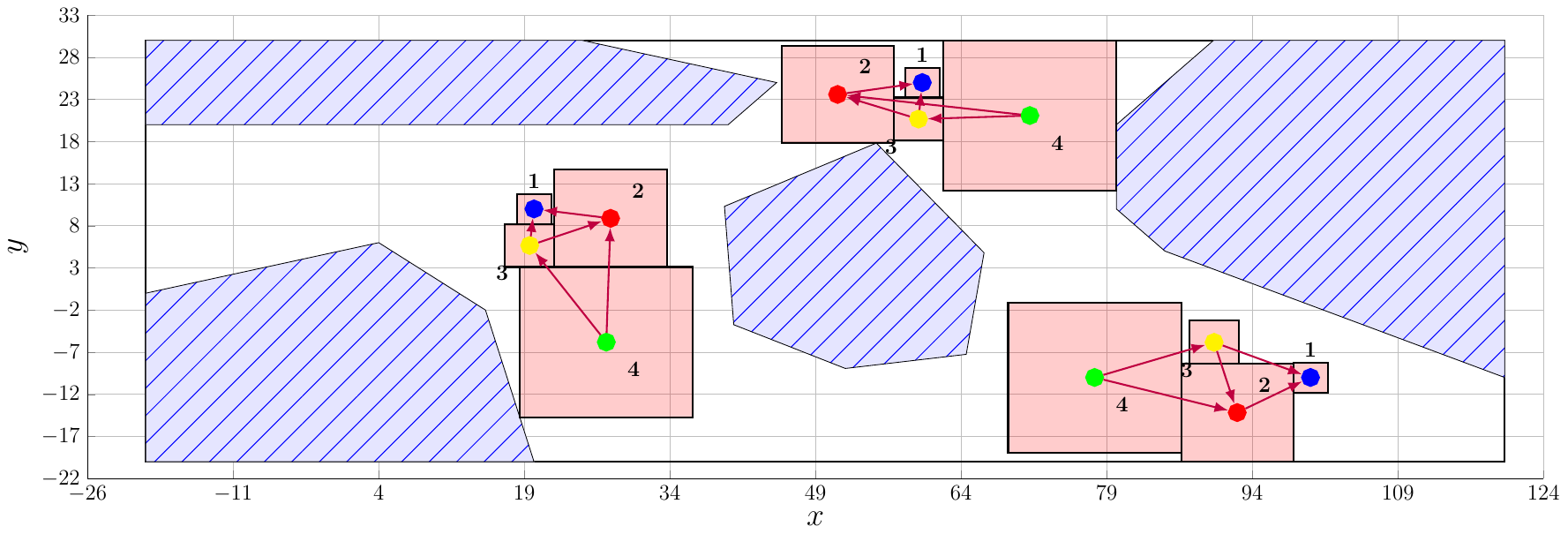}
  \caption{Tight formation in a world with obstacles}
  \label{fig:tight_formation}
\end{figure*}

These requirements are expressed as the following nonlinear optimization problem, which solves for the displacement components $z_{ij}^\star$, collected in the column vector $z^\star$.
\begin{prop}
\label{prop:optimization}
Under dynamics \eqref{eq:noisy_tracking_error} and with requirements~\ref{item:opt-1}--\ref{item:opt-4}, the target displacements $z^\star$ which characterize the tightest displacement formation result from
\begin{subequations}
\label{eq:optimization}
    \begin{align}
        \label{eq:optimization_a} \bar{z}^\star &= \argmin_{z^\star}\: \left\| z^\star\right\|_{Q_z}\\
        \label{eq:optimization_b}\text{s.t.}\:& \begin{bmatrix}\mathbf {e}_1^\top \otimes I_n\end{bmatrix}\begin{bmatrix}\mathbf {e}_1^\top \otimes I_n\\\tilde H\end{bmatrix}^{-1}\begin{bmatrix}\bar p\\ z^\star\end{bmatrix} = {\bar p},\\
         \label{eq:optimization_c}&|z^\star|\not\leq \left|\tilde H\right| \cdot \left| S_p \tilde{V}_\Gamma \tilde{b}_\Gamma\right|,\\
        \label{eq:optimization_d}& \begin{bmatrix}\mathbf {e}_1^\top \otimes I_n\\\tilde H\end{bmatrix}^{-1}\begin{bmatrix}\bar p\\ z^\star\end{bmatrix}\notin \bigl(\mathcal O\times \ldots \times \mathcal O\bigr)\mkern-0mu \oplus B_{\Gamma, UB},
    \end{align}
\end{subequations}
where $\mathbf{e}_1$ is the first column from $I_N$, $S_p=\begin{bmatrix} I_{nN}&0_{nN}\end{bmatrix}$ selects the row indices which correspond to the position component of the state, weighting matrix $Q_z$ is strictly positive definite and of appropriate dimensions and ``$\times$'' denotes the Cartesian product between sets.  
\end{prop}
\begin{proof}
Note that~\eqref{eq:optimization} assumes quasi-stationary behavior, i.e., $\zeta = 0$, as required by item~\ref{item:opt-1}. Recall that $p^\star$, the vector of vertically stacked target positions, satisfies $z^\star = \tilde{H} p^\star$. Appending the anchoring condition $p_1^\star = \bar{p}$, equivalently written as $\bigl[\mathbf {e}_1^\top \otimes I_n\bigr]p^\star = \bar{p}$, anchors the formation to a specific configuration. This allows $p^\star$ to be expressed in terms of the displacement vector $\mathbf{z}^\star$ and the leader’s target $\bar{p}$ as
\begin{equation}
\label{eq:pstar}
     p^\star=\begin{bmatrix}\mathbf {e}_1^\top \otimes I_n\\\tilde H\end{bmatrix}^{-1}\begin{bmatrix}\bar p\\ z^\star\end{bmatrix}.
\end{equation}
Fixing the leader’s position (item~\ref{item:opt-2}) is thus achieved through constraint~\eqref{eq:optimization_b}. The tightness requirement in item~\ref{item:opt-3} is enforced by penalizing the displacement magnitude in the cost function~\eqref{eq:optimization_a}, while constraint~\eqref{eq:optimization_c} ensures safe inter-agent spacing.  Vector $\begin{bmatrix} e_p^\top &  e_v^\top \end{bmatrix}^\top$ is bounded elementwise by $B_{\Gamma,\mathrm{UB}}$; taking only the position component and adapting~\eqref{eq:ub_box} yields
\begin{equation}
\label{eq:ep_bound}
    |e_p|\leq S_p \tilde{V}_\Gamma \tilde{b}_\Gamma.
\end{equation}
The relative displacement between agents $i$ and $j$ satisfies $z_{ij}= p_i-p_j= (p_i-p_i^\star) - (p_j-p_j^\star) + p_i^\star - p_j^\star = e_{p_i}-e_{p_j}+z_{ij}^\star$. Noting that a sufficient condition to avoid inter-agent collision is to enforce $z_{ij} \neq 0$ and assuming that $|e_{p_i}-e_{p_j}| \leq \epsilon$ with $\epsilon \in \re{}_+$, we can readily rewrite the sufficient condition as  $|z_{ij}^\star| \not\leq \epsilon$. Applying this to all agent pairs leads to implication: 
\begin{equation}
\label{eq:zstar_implication}
\bigl(| z^\star| \not\leq |\tilde H|\cdot |e_p|\bigr) \implies \bigl(z \neq \mathbf{0}_{nM}\bigr).
\end{equation}
Checking ``$x\not\leq y$'' means that there exist at least an index $i$ such that $x_i>y_i$.
Combining \eqref{eq:zstar_implication} with~\eqref{eq:ep_bound} and noting that $|\tilde H  e_p|\leq |\tilde H|\cdot | e_p|$ provides ~\eqref{eq:optimization_c}. Finally, overall obstacle avoidance (item~\ref{item:opt-4}) is guaranteed by the exclusion constraint~\eqref{eq:optimization_d}, where each obstacle set is inflated by the bounding hyper-rectangle $B_{\Gamma,\mathrm{UB}}$, concluding the proof.
\qed    
\end{proof}

\begin{rem}
The constraint~\eqref{eq:optimization_b} can be directly written as in \ref{item:opt-2}, but the goal is to keep the formulation in terms of displacement-based formation description, by using relative position/velocity vectors.
\eor
\end{rem}
\section{Results}
\label{sec:results}
To test the theoretical results from Prop.~\ref{prop:optimization}, we considered a LFF formation with $N=4$ agents, as described in the illustrative example from Sec.~\ref{sec:prerequisites}.

In Figure~\ref{fig:tight_formation}, we first defined the obstacles (dashed blue), then the agents: \textit{leader} (\textbf{1} \begin{tikzpicture}[baseline=-2pt]\node[circle,draw,minimum width=2pt,line width=1pt,blue,fill=blue] at(0.4,0){};\end{tikzpicture}), \textit{first-follower} (\textbf{2} \begin{tikzpicture}[baseline=-2pt]\node[circle,draw,minimum width=2pt,line width=1pt,red,fill=red] at(0.4,0){};\end{tikzpicture}), \textit{second-follower} (\textbf{3} \begin{tikzpicture}[baseline=-2pt]\node[circle,draw,minimum width=2pt,line width=1pt,yellow,fill=yellow] at(0.4,0){};\end{tikzpicture}), \textit{third-follower} (\textbf{4} \begin{tikzpicture}[baseline=-2pt]\node[circle,draw,minimum width=2pt,line width=1pt,green,fill=green] at(0.4,0){};\end{tikzpicture}). The edges of the interaction graph are depicted using arrows (\begin{tikzpicture}[baseline=-2pt]\draw[line width=1pt, purple, -Latex] (0,0) -- (0.6,0);\end{tikzpicture}).

Minimizing the volume as in Lemma~\ref{lem:control_synthesis} we obtain $k_p=0.31, k_v=3.15$ which correspond to the volume value of $\vol\Omega_{\Gamma,\mathrm{UB}}(k_p,k_v) = 14.73$. Introduced back into~\eqref{eq:eigenvectors_gamma}, we get the bound $|\tilde V_\Gamma|\tilde b_\Gamma$, {which allows  to characterize the safety margins around each of the target positions.}

Given the inherently non-convex nature of the obstacle avoidance constraint~\eqref{eq:optimization_d}, a MIP-based representation \cite{prodan_mixed-integer_2016} is necessary (and utilized) for its implementation and, implicitly, resolution of \eqref{eq:optimization}.

To illustrate the richness of the possible formation configurations, we considered three anchoring positions for the leader
\begin{equation*}
    \bar p_1 \in \bigg\{\begin{bmatrix}10\\20\end{bmatrix}, \begin{bmatrix}
        100 \\ -10\end{bmatrix}, \begin{bmatrix}
            60	\\ 25\end{bmatrix}\bigg\},
\end{equation*}
chosen such that the agents are constricted by either the world's boundary or by some of the obstacles. 

Using $|\tilde V_\Gamma|\tilde b_\Gamma$ calculated prior, the optimization problem~\eqref{eq:optimization} is solved for each of the leader position, giving three steady-state formations, as illustrated in Fig.~\ref{fig:tight_formation}. We observe that the results are formation instances which are tightly arranged around the leader while respecting requirements \ref{item:opt-1}--\ref{item:opt-4}, as expected. 

We performed the simulation in MATLAB environment, under Windows 11 operating system. The implementation of Lemma~\ref{lem:control_synthesis} was done using CasADi \cite{andersson2019casadi}. The problem described in Proposition~\ref{prop:optimization} was implemented in YALMIP \cite{Lofberg2004}, and solved using the Gurobi solver \cite{gurobi}. The code for both is stored in the sub-directory ``/replan-public/formation-control/ECC-2026'' of the Gitlab repository found at: \url{https://gitlab.com/replan/replan-public.git}.


\section{Conclusions}
\label{sec:conclusions}


This work addressed controller synthesis for displacement-based formation control from a set-theoretic perspective to handle bounded disturbances from measurement noise. By applying set invariance theory, we developed a rigorous framework for guaranteeing performance under persistent disturbances and enabling optimal selection of control parameters to keep formation errors within prescribed bounds. The approach was validated on the challenging task of maintaining a tight formation in a multi-obstacle environment, demonstrating its ability to ensure safety and formation integrity in constrained settings.

Future work will extend this framework to nonlinear distance-based formation control and to more realistic, non-holonomic agent models such as unicycle dynamics. Moreover, the insights gained from solving the tight-formation problem will serve as the foundation for a high-level motion planning strategy for multi-agent systems.

\appendix
\subsection{Proof of Corollary~\ref{cor:volume_proof}}
\label{proof:cor_volume}


We reformulate \eqref{eq:vol_ub} using the matrices \eqref{eq:eigenvalues_gamma}--\eqref{eq:eigenvectors_gamma}. By the Kronecker determinant's property~\cite[Ex.~4.2.1]{horn1994topics}, we have that 
\begin{align} \label{eq:det_V_gamma}
    \nonumber\det\tilde V_\Gamma&=\left(\det\begin{bmatrix}
V & V \\
V\Lambda_+ & V\Lambda_-
\end{bmatrix}\right)^n\\
\nonumber &=\left(\left[\det V\right]^2\cdot \det (\Lambda_--\Lambda_+)\right)^n\\
&=\left(\left[\det V\right]^2\cdot \prod_{i=1}^N(-\sqrt{\Delta_i})\right)^n.
\end{align}

Further, term $\tilde b_\Gamma$, defined as in \eqref{eq:notation_rhs_ub} for dynamics \eqref{eq:tracking_error}, becomes
\begin{equation}
\label{eq:b_gamma}    \tilde{b}_\Gamma=\bigl|\bigl[\real(\tilde{\Lambda}_\Gamma)\bigr]^{-1}\bigr|
        \cdot \left(|\tilde{V}^{-1}_\Gamma c| + |\tilde{V}^{-1}_\Gamma G|\cdot \mathbf 1_{2nN}\right).
\end{equation}
As a first step we note that $\tilde V_\Gamma= V_\Gamma \otimes I_n$ which implies that $\tilde{V}_{\Gamma}^{-1}= V_{\Gamma}^{-1} \otimes I_n$. Since $V$ is invertible, for the inverse of $V_{\Gamma}$, we use the Schur complement for 2x2 block matrices, which results in:
\begin{equation}    
\label{eq:v_gamma_inv}
 \tilde{V}_{\Gamma}^{-1} = \begin{bmatrix}
                   V^{-1}+S^{-1}\Lambda_+V^{-1} & -S^{-1}V^{-1} \\
                   -S^{-1}\Lambda_+V^{-1} & S^{-1}V^{-1}
                   \end{bmatrix} \otimes I_n
\end{equation}
where $S = \Lambda_--\Lambda_+$. 

From \eqref{eq:b_gamma} and denoting $\bar \delta = |c|+|G|\mathbf 1_{2nN}$, we have that 
$b_{\Gamma}=\bigl|\tilde{\Lambda}_{\Gamma}^{-1}\bigr|\cdot |\tilde{V_\Gamma}^{-1}|\overline{\delta}$, where we have the aim to express both $\Lambda_\Gamma^{-1}$ and $V_\Gamma^{-1}$ in terms of the parameters $k_p,k_v$. First, under the closed-loop stability assumption that $\mathrm{eig}(\Lambda)\in \mathbb R_-^n$ and using the Kronecker property~\ref{item:k3}, it follows that 
\begin{equation}
\label{eq:lambda_gamma_inv}
    \bigl|\tilde{\Lambda}_{\Gamma}^{-1}\bigr|\mkern-4mu=\mkern-4mu\left|\mkern-4mu
\begin{bmatrix} \Lambda_+^{-1} & O_N \\ O_N & \Lambda_-^{-1}\end{bmatrix}\otimes I_n\mkern-4mu\right|\mkern-4mu = \mkern-6mu\underbrace{\begin{bmatrix} -\Lambda_+^{-1} & \hphantom{-}O_N \\ \hphantom{-}O_N & -\Lambda_-^{-1}\end{bmatrix}}_{ \Lambda^{-1}_{\Gamma}}\mkern-2mu\otimes I_n.
\end{equation}
Second, recalling \eqref{eq:v_gamma_inv}, we have
\begin{equation}
\label{eq:v_gamma_tilde}
\tilde{V}_\Gamma^{-1}= \biggl(\underbrace{
\begin{bmatrix} S^{-1} & O_N \\ O_N & S^{-1} \end{bmatrix}
\begin{bmatrix} \hphantom{-}\Lambda_- & -I_N \\ -\Lambda_+ & \hphantom{-}I_N \end{bmatrix}
\begin{bmatrix} V^{-1} & O_N \\ O_N & V^{-1} \end{bmatrix}}_{V_{\Gamma}^{-1}}\biggr)\otimes I_n.
\end{equation}
Adding \eqref{eq:lambda_gamma_inv} and \eqref{eq:v_gamma_tilde} back into \eqref{eq:b_gamma} and recalling the Kronecker properties, we have that

\begin{align*}
    \tilde{b}_{\Gamma} &\overset{\ref{item:k3}}{=} \biggl(\begin{bmatrix} -\Lambda_+^{-1} & O_N \\ O_N & -\Lambda_-^{-1}\end{bmatrix} \otimes I_n\biggr)(|V_{\Gamma}^{-1}| \otimes I_n)\overline{\delta} \\
    &\overset{\ref{item:k1}}{=}\biggl[\biggl(\begin{bmatrix} -\Lambda_+^{-1} & O_N \\ O_N & -\Lambda_-^{-1}\end{bmatrix}|V_{\Gamma}^{-1}|\biggr)\otimes I_n\biggr]\biggl(\overline{\delta}_{2N}\otimes \mathbf{1}_n\biggr) \\
    &\overset{\ref{item:k1}}{=}\biggl(\underbrace{\begin{bmatrix} -\Lambda_+^{-1} & O_N \\ O_N & -\Lambda_-^{-1}\end{bmatrix}|V_{\Gamma}^{-1}|\:\overline{\delta}_{2N}}_{b_{ \Gamma}}\biggr)\otimes(I_n\mathbf{1}_n)  
\end{align*}
If $A$ is a diagonal matrix, then $A|B|=|AB|$ for any matrix $B$. Hence, $b_\Gamma$ may be rewritten as
\begin{multline}
b_{\Gamma}=\left|
-\Lambda_{\tilde{\Gamma}}^{-1}
\begin{bmatrix} S^{-1} & O_N \\ O_N & S^{-1} \end{bmatrix}
\begin{bmatrix} \hphantom{-}\Lambda_- & -I_N \\ -\Lambda_+ & \hphantom{-}I_N \end{bmatrix}
\begin{bmatrix} V^{-1} & O_N \\ O_N & V^{-1} \end{bmatrix}
\right|\overline{\delta}_{2N},\\
=\left|\begin{bmatrix} 
D_{11} & D_{12} \\
D_{21} & D_{22}
\end{bmatrix}
\begin{bmatrix} V^{-1} & O_N \\ O_N & V^{-1} \end{bmatrix}
\right|\overline{\delta}_{2N}
\end{multline}
\noindent with shorthand notation
\begin{align*}
D_{11} &= -[(\Lambda_--\Lambda_+)\Lambda_+]^{-1}\Lambda_- \mkern-12mu&=&\diag\biggl(\frac{1}{\sqrt{\Delta_i}}
\frac{k_v \lambda_i + \sqrt{\Delta_i}}{k_v \lambda_i - \sqrt{\Delta_i}}\biggr),\\
D_{12} &= [(\Lambda_--\Lambda_+)\Lambda_+]^{-1}\mkern-12mu&=&\diag\biggl(-\frac{2}{\sqrt{\Delta_i}}
\frac{1}{k_v \lambda_i - \sqrt{\Delta_i}}\biggr), \\
D_{21} &= [(\Lambda_--\Lambda_+)\Lambda_-]^{-1}\Lambda_+\mkern-12mu&=&\diag\biggl(\frac{1}{\sqrt{\Delta_i}}
\frac{\sqrt{\Delta_i} - k_v \lambda_i}{k_v \lambda_i + \sqrt{\Delta_i}}\biggr), \\
D_{22} &= -[(\Lambda_--\Lambda_+)\Lambda_-]^{-1}\mkern-12mu&=&\diag\biggl(-\frac{2}{\sqrt{\Delta_i}}
\frac{1}{k_v \lambda_i + \sqrt{\Delta_i}}\biggr),
\end{align*}
and where $\Delta_i=k_v^2\lambda_i^2-4k_p\lambda_i$. With these we can express $b_{\Gamma}$ explicitly in terms of $k_p,k_v$:
\begin{align*}
b_{\Gamma}=\begin{bmatrix}
(|D_{11}V^{-1}|+|D_{12}V^{-1}|)\overline{\delta}_N \\
(|D_{21}V^{-1}|+|D_{22}V^{-1}|)\overline{\delta}_N
\end{bmatrix} = \begin{bmatrix}
(|D_{11}|+|D_{12}|)\overline{d} \\
(|D_{21}|+|D_{22}|)\overline{d}
\end{bmatrix}
\end{align*}
where $\overline{d}=V^{-1}\overline{\delta}_N$. Taken elementwise, we have
\begin{equation}
(b_{\Gamma})_{i}=
\begin{cases}
\frac{1}{\sqrt{\Delta_i}}\frac{2+k_v\lambda_i+\sqrt{\Delta_i}}{k_v\lambda_i-\sqrt{\Delta_i}}\overline{d}_i, &\text{if } i \in \overline{1,N} \\
\frac{1}{\sqrt{\Delta_i}}\frac{2+k_v\lambda_i-\sqrt{\Delta_i}}{k_v\lambda_i+\sqrt{\Delta_i}}\overline{d}_{i-N}, &\text{if } i \in \overline{N+1,2N}
\end{cases}
\end{equation}
Recalling that $\tilde{b}_{\Gamma}=b_{\Gamma} \otimes \mathbf{1}_n$ leads to
\begin{equation}
\label{eq:b_gamma_prod}
\prod_{i=1}^{2nN}{(\tilde{b}_{\Gamma}})_i=\biggl(\prod_{i=1}^{2N}(b_{\Gamma})_i\biggr)^n=\prod_{i=1}^Nd_i^2\begin{bmatrix}\frac{1+\lambda_i(k_p+k_v)}{\sqrt{\Delta_i}k_p\lambda_i}\end{bmatrix}^n.
\end{equation}
Introducing \eqref{eq:det_V_gamma} and \eqref{eq:b_gamma_prod} back into \eqref{eq:volume_ub} we obtain the volume, as in \eqref{eq:vol_ub}, thus concluding the proof. 

\end{document}